\documentclass[USenglish, a4paper, thm-restate,numberwithinsect, cleveref]{lipics-v2021} 
\usepackage[dvipsnames,svgnames,x11names]{xcolor}
\usepackage[utf8]{inputenc}
\usepackage{tabularx}
\usepackage{amsthm}
\usepackage{amsmath}
\usepackage{amssymb}
\usepackage{amsfonts}
\usepackage{mathtools}
\usepackage{enumerate}
\usepackage[sort,numbers]{natbib}
\usepackage{tikz}
\usepackage{xspace}
\usepackage{todonotes}
\usepackage{comment}

\usepackage{multirow}
\usepackage{booktabs}
\usepackage{makecell}
\usepackage{caption}

\usepackage{tikz}

\usetikzlibrary{arrows,decorations.pathmorphing,decorations.pathreplacing,backgrounds,positioning,fit,matrix}
\usetikzlibrary{shapes,calc,patterns,arrows.meta}
\tikzset{
	vert/.style={circle,inner sep=1.5,fill=white,draw=black,minimum size=.3cm},
	vert2/.style={inner sep=1.5,fill=white,draw=black,minimum size=.3cm},
    dummy/.style={circle,fill=black,draw=black,inner sep=2.5},
	edge/.style={color=black, line width=1pt},
	diredge/.style={->,>={Stealth[width=8pt,length=8pt]},color=black, line width=1pt},
	timelabel/.style={circle,inner sep=1.2,fill=white,font=\footnotesize, text centered},
	timelabel2/.style={circle,inner sep=1.2,fill=lightgray!40!white,font=\footnotesize, text centered},
	wave/.style={decorate,decoration={coil,aspect=0}},
	dirwave/.style={->, >={Stealth[width=8pt,length=8pt]},decorate,decoration={coil,aspect=0}},
	diredge2/.style={->,>={Stealth[width=8pt,length=8pt]}}
}

\theoremstyle{definition}

\newtheorem{algorithm}{Algorithm}
\usepackage{soul}

\crefname{figure}{Figure}{Figures}

\newcommand{\commentout}[1]{}

\newcommand{\problemdef}[3]{
\noindent ~~\fbox{\parbox{12cm}{
      #1\medskip\\
      \textrm{Input:}     #2 \\
      \textrm{Question:}  #3
    }
  }\smallskip
}

\newcommand{\G}{\ensuremath{\mathcal{G}}\xspace}
\newcommand{\TC}{\textsf{TC}\xspace}
\newcommand{\MTS}{\textsc{Minimum Temporal Spanner}\xspace}
\DeclareMathOperator*{\argmax}{arg\,max}

\pdfoutput=1
\hideLIPIcs
\nolinenumbers

\title{Minimum Temporal Spanners in Happy Graphs}

\author{Arnaud Casteigts}{Department of Computer Science, University of Geneva, Switzerland}{arnaud.casteigts@unige.ch}{https://orcid.org/0000-0002-7819-7013}{Supported by the French ANR project TEMPOGRAL (ANR-22-CE48-0001) and Swiss NSF project RECAPT (200021-236640).}

\author{Michelle D\"oring}{Hasso Plattner Institut, Potsdam, Germany \and \url{https://michelledoering.notion.site/}}{michelle.doering@hpi.de}{https://orcid.org/0000-0001-7737-3903}
    {HPI Research School on Foundations of AI (FAI)}

\author{Hendrik Molter}{Chair for Algorithm Engineering, Hasso Plattner Institute, Potsdam, Germany \and Department of Computer Science, Ben-Gurion~University~of~the~Negev, 
	Beer-Sheva, 
	Israel}{hendrik.molter@hpi.de}{https://orcid.org/0000-0002-4590-798X}{Supported by the Israel Science Foundation, grant nr.~1470/24, by the European Union's Horizon Europe research and innovation programme under grant agreement 949707, by the European Research Council, grant nr.~101039913 (PARAPATH), and by the Deutsche Forschungsgemeinschaft (DFG, German Research Foundation) -- 565381415.}

\author{Meirav Zehavi}{Department of Computer Science, Ben-Gurion~University~of~the~Negev, Beer-Sheva, Israel}{meiravze@bgu.ac.il}{https://orcid.org/0000-0002-3636-5322}{Supported by the Israel Science Foundation, grant nr.~1470/24, and by the European Research Council, grant nr.~101039913 (PARAPATH).}

\authorrunning{Arnaud Casteigts, Michelle D\"oring, Hendrik Molter, Meirav Zehavi}
\Copyright{Arnaud Casteigts, Michelle D\"oring, Hendrik Molter, Meirav Zehavi}



\begin{document}
\maketitle

% \author{Arnaud~Casteigts}
% \author{Hendrik~Molter\thanks{Supported by the ERC, grant number 949707.}}
% \author{Meirav~Zehavi}

% \affil{\small Department of Computer Science, University of Geneva, Switzerland\\ \texttt{arnaud.casteigts@unige.ch}\\
% \small Department of Computer Science, Ben-Gurion~University~of~the~Negev, 
% Beer-Sheva, 
% Israel\\ \texttt{molterh@post.bgu.ac.il, meiravze@bgu.ac.il}}

% \date{}

\begin{abstract}
Temporal graphs are graphs whose edge set changes over discrete time steps. These graphs can model scenarios related to transportation, scheduling, social networks, and networking algorithms. A temporal graph is temporally connected (TC) if all pairs of vertices can reach each other using paths that traverse the edges in a time-respecting way (temporal paths). Given a TC graph, a natural problem is to find a small spanning subgraph that preserves temporal connectivity, analog to the notion of a spanning tree in static graphs. These structures, known as temporal spanners, are fundamental and their properties have attracted significant attention in the past decade. In particular, the problem of minimizing the size (number of edges or number of edge appearances) of a temporal spanner was shown to be intractable for several incomparable settings of temporal graphs, leaving important cases open.

In this article, we unify and complement these results by showing that this problem is NP-hard even on temporal graphs that are simple and proper (also known as ``happy''), i.e., where every edge appears only one time, and a vertex cannot be incident to several edges simultaneously. Proving hardness in this extremely restricted setting implies, at once, that the two versions of the problem (edges and edge appearances) are NP-hard in all the standard temporal graph settings. This unified result being obtained, we initiate the parameterized study of this problem, showing that in the happy setting, the problem can be solved in polynomial time if the underlying graph has a constant-size vertex cover. This result is, to the best of our knowledge, the first positive result on temporal spanners in general. We also show that it is essentially best possible, in the sense that the problem is W[1]-hard when parameterized by the vertex cover number of the underlying graph.
\end{abstract}

\section{Introduction}

Minimum connected spanning subgraphs, also known as \emph{spanning trees}, are one of the most fundamental primitives in algorithmic graph theory. The development of algorithms to compute these structures goes back to the early 20th century and generated an important sequence of results over the century. Indeed, spanning trees serve as an important subroutine for a multitude of problems in various application areas.

\emph{Temporal graphs} are a canonical generalization of (static) graphs, which captures changes over time. In most models, the vertex set remains unchanged while for every (discrete) time step in some finite lifetime, the edge set may be different.
The incorporation of time adds a new layer of complexity to connectivity-related problems: paths in a temporal graph need to respect time, that is, they have to use edges at non-decreasing (or increasing) time steps. Reachability based on temporal paths is neither symmetric nor transitive, a major difference to the static setting. For this reason, reachability-related problems have attracted significant attention in temporal graphs, where they are typically much harder computationally than their static counterparts.

The present work is dedicated to the problem of computing so-called \emph{minimum temporal spanners}, which can be seen as the temporal analog of minimum spanning trees. A temporal graph is \emph{temporally connected} (\TC) if all pairs of vertices can reach each other via temporal paths. The goal is to compute, given a \TC\ temporal graph, a minimum-size spanning subgraph that preserves \TC (\MTS). From a scheduling perspective, this problem corresponds for instance to scenarios where a pre-determined communication schedule is given and one must simplify it by selecting the smallest subset of scheduled communications that preserves all-pairs causality. 

This problem is known to be NP-hard in several incomparable settings~\cite{AkridaGMS17,AF16,CC23}, and a considerable amount of research has been dedicated to finding structural lower and upper bounds for the size of minimum temporal spanners~\cite{AF16,CasteigtsPS21,KKK02,AkridaGMS17,angrick2024linear,CSS23,BalevSS24,CCC25} (more details below). Observe that research in this area mostly considers unweighted graphs and the minimization criterion concerns the \emph{size} of the spanner, either in terms of the number of underlying edges\footnote{The \emph{underlying graph} of a temporal graph is a static graph that contains all the edges that appear at least once.} (\textsc{Min-Edge} version) or in terms of the number of edge appearances (\textsc{Min-Label} version). The reason for this focus is that the minimum size of a temporal spanner is not universal, unlike for static graphs, where spanning trees of size $n-1$ always exist when the graph is connected.

\subparagraph*{Known Results.}
The hardness of \MTS has been studied for a decade or so, starting with \citet{AkridaGMS17} and \citet{AF16} who showed independently that the problem is APX-hard (and a fortiori, NP-hard). Interestingly, a close examination of the corresponding reductions shows that both results are actually incomparable. 
\citet{AF16} consider the \emph{non-strict} setting, where temporal paths are allowed to use edges along non-decreasing time steps. Moreover, the temporal graph in their case is \emph{simple}, which means that every edge appears only at a single time step. As such, their result holds for both the \textsc{Min-Edge} and \textsc{Min-Label} versions. In contrast, 
% and the total number of time steps is upper-bounded by a constant.
\citet{AkridaGMS17} consider the \emph{strict} setting, where temporal paths must use edges along increasing time steps. The temporal graphs used in their reduction are not simple, and their result holds only for the \textsc{Min-Label} version. These results are incomparable because the strict and the non-strict settings are themselves incomparable in terms of expressivity~\cite{CasteigtsCS24}. More recently, \citet{CC23} showed that deciding if a \TC graph admits a temporal spanner whose underlying graph is a tree is also NP-hard on non-simple temporal graphs, which implies (again) that the \textsc{Min-Edge} version of \MTS is NP-hard. However, this time, the result holds for both strict and non-strict temporal paths, due to a reduction that considers \emph{proper} temporal graphs, where the vertices are incident to at most one edge at a time (so the distinction between strict and non-strict paths does not exist). As it turns out, each of these three results establishes NP-hardness in cases that are not covered by the two others, making the three results incomparable. In addition, the combination of \textsc{Min-Edge} \MTS in the strict setting was left open so far.

\subparagraph*{Contributions.}
Our first contribution is to present a new hardness result that unifies and extends the above NP-hardness results, establishing \emph{at once} that \MTS is NP-hard for both the \textsc{Min-Edge} and the \textsc{Min-Label} versions in both the strict and the non-strict settings. This result, presented in Section~\ref{sec:hardness}, is obtained by means of a reduction to temporal graphs that are simultaneously simple and proper, also known as \emph{happy} temporal graphs, where both settings and both versions of the problem coincide into a single unified problem. Doing so, we close a long series of investigations on the hardness of \MTS and resolve Open Question~4 in~\cite{CasteigtsCS24}, asking whether such a unification was possible. The situation is summarized in \cref{fig:table}.

% \begin{figure}[h]
%   \centering
% \begin{tabular}{|c|c|c|c|}
%   \hline
%   Article & Setting & Status & Min edges/labels \\\hline
%   \cite{AF16}&  ``simple \& non-proper \& non-strict'' & APX-hard & Both\\
%   \cite{AkridaGMS17}& ``non-simple \& non-proper \& strict'' & APX-hard & Labels \\
%   \cite{CC23}& ``non-simple \& proper'' & NP-hard & Edges \\
% %  Here& ``simple \& non-proper \& strict'' & NP-hard & Both \\
%   \textbf{Here}& \textbf{``simple \& proper'' (happy)} & \textbf{NP-hard} & \textbf{Both}\\
%   \hline
% \end{tabular}
% \caption{\label{fig:table}Existing and new results on the hardness of \MTS.}
% \end{figure}

\begin{figure}[h]
  \centering
\begin{tabular}{|c|c||c|c|}
  \hline
  Article & Reduction to & Setting & Version \\\hline
  \cite{AF16}& simple \& not proper & non-strict & both\\
  \cite{AkridaGMS17}& not simple \& not proper& strict  & \textsc{Min-Label} \\
  \cite{CC23}& not simple \& proper & both  & \textsc{Min-Edge} \\\hline
%  Here& simple \& not proper &strict & both \\
  This paper& simple \& proper (happy) & \textbf{both} & \textbf{both}\\
  \hline
\end{tabular}
\caption{\label{fig:table}Existing and new NP-hardness results for \MTS.}
\end{figure}

The NP-hardness of \MTS being clarified and made more general, our second contribution is to initiate the parameterized study of this problem in the case of happy temporal graphs. It turns out that our reduction not only establishes NP-hardness of \MTS, it also establishes that this problem is W[1]-hard when parameterized by the vertex cover number (size of a minimum vertex cover) of the underlying graph. Thus, under widely believed assumptions (namely, FPT $\ne$ W[1]), no fixed-parameter tractable algorithm exists for this parameter, nor a fortiori for smaller parameters like \emph{treewidth}. In other words, this problem is very hard already in the special case of happy temporal graphs.

On the positive side, we show that \MTS can still be solved in polynomial time if the vertex cover number has constant size (namely, it is in XP when parameterized by the vertex cover number of the underlying graph), leaving some hope to solve it efficiently in special cases (for example, in split graphs with a small clique). To the best of our knowledge, our XP algorithm (presented in Section~\ref{sec:xp-vc}) is the first nontrivial \emph{positive} result for \MTS, apart from the very special case that the input graph itself is a tree~\cite{AF16}. To obtain the algorithm, we reveal insights into the structure of temporal spanners (the VC-tree lemma), which we believe is of general interest for further connectivity-related temporal graph problems.

\subparagraph*{Further Related Work.} First of all, let us stress that the results from~\citet{AkridaGMS17} and \citet{AF16} are APX-hardness results. As such, our results do not replace these results entirely, they only replace what these results imply for the NP-hardness of the problem. In addition, 
\citet{AF16} give approximation algorithms and further approximation lower bounds for an edge-weighted version of the problem. They also give a polynomial-time algorithm for the weighted problem when the underlying graph is a tree (mentioned in the previous paragraph). 
\citet{CC23} further study the problem of finding a spanner that connects all vertex pairs by so-called bi-paths (a pair of a temporal path from $u$ to $v$ and a temporal path from $v$ to $u$ that use the same underlying edges) and show that this problem can be solved in polynomial time without restriction on the size, but finding a bi-spanner with at most $k$ edges is NP-hard.

On the structural side, \citet{KKK02} observed that \TC graphs do not always admit temporal spanners whose underlying graph is a tree. They also showed that a classical construction from gossip theory (a specific labeling of hypercubes) implies that there are some \TC graphs on $\Theta(n \log n)$ edges, all of which are necessary for preserving \TC. \citet{AF16} further strengthened this to showing that there are \TC graphs on $\Theta(n^2)$ edges, all of which are also necessary. Thus, small-size spanners do not exist unconditionally in temporal graphs. However, they are guaranteed if the underlying graph is a complete graph~\cite{CasteigtsPS21} or if the temporal graph is a temporal analog of Erdös-Renyi graphs~\cite{CasteigtsRRZ21}. Beyond the size, \citet{BiloDG0R22} study a version of the problem where the temporal spanners should preserve original distances, obtaining a general tradeoff between the number of edges and the stretch in the case of complete temporal graphs. Finally, various constructions of temporal trees that preserve only reachability from a given vertex to all others can be found in~\cite{KKK02,XFJ03,huang2015minimum}.

There are several more loosely related problems that have been investigated. For example, the problem of finding a labeling for a given static graph, such that the resulting temporal graph is \TC has been studied in various settings~\cite{gobel1991label,MertziosMS19,KlobasMMS22}. This can be interpreted as finding a spanner for a temporal graph where every edge appears at every time step. Similarly, \citet{bellitto2025temporal} consider the problem of augmenting a partial temporal graph with a minimum set of additional time edges from a list of candidates edges in order to make the graph \TC (which is equivalent to the minimum spanner problem if the partial graph is empty and the list contains all the edges). Furthermore, in the static setting, computing a spanning tree or a feedback edge set are essentially equivalent problems. In the temporal setting, the problem of computing a temporal feedback edge set behaves quite different than the problem of computing minimum spanners~\cite{HaagMNR22}. To the best of our knowledge, it is not known whether there are any interesting relations between the two problems.

% There has also been work on the problem of finding a minimum temporal subgraph where a given vertex can reach all other vertices. \citet{KKK02} showed that this problem can be solved in polynomial time. \citet{AF16} and \citet{huang2015minimum} showed that the edge-weighted version of the problem is NP-hard and hard to approximate. \citet{AF16} further give an XP-algorithm for the treewidth of the underlying graph as a parameter.\citet{huang2015minimum} showed that a minimum spanner can be computed in polynomial time (for the unweighted version) that preserves earliest arrival times.

\section{Preliminaries and Problem Definition}\label{sec:prelims}

A \emph{temporal graph} can be represented as a triple $\mathcal{G}=(V,\mathcal{E},T)$ consisting of a set of vertices $V$, a set of \emph{time edges} $\mathcal{E}\subseteq\binom{V}{2}\times [T]$, and a maximum time label $T$.
We assume that there is $(\{u,v\},t)\in\mathcal{E}$ such that $t=T$.
The \emph{underlying graph} of $\mathcal{G}$ is the static graph $G_U=(V,E_U)$ with $E_U = \{\{u,v\}\mid \exists t\in[T] \text{ s.t.\ }(\{u,v\},t)\in\mathcal{E}\}$. It is also common to think of the temporal graph as a labeled graph, where the edges of the underlying graph $G_U$ are labeled with their presence times.

Given a temporal graph $\mathcal{G}=(V,\mathcal{E},T)$ and a vertex subset $V'\subseteq V$, we define $\mathcal{G}[V']:=(V',\mathcal{E}',T')$, where $\mathcal{E'}=\{(\{u,v\},t)\mid (\{u,v\},t)\in\mathcal{E}\wedge \{u,v\}\subseteq V'\}$, and $T'$ is the maximum label among time edges in $\mathcal{E}'$. Furthermore, given a temporal graph $\mathcal{G}=(V,\mathcal{E},T)$ and a time edge $(\{u,v\},t)$, we define $\mathcal{G}+(\{u,v\},t):=(V\cup\{u,v\},\mathcal{E}\cup\{(\{u,v\},t)\}, T')$, where $T'=\max(T,t)$.

A temporal graph is \emph{simple} if for each edge $e$ of the underlying graph, there is at most one time label $t$ such that $(e,t)\in\mathcal{E}$. We say that a temporal graph is \emph{proper}, if for each pair of time edges $(e,t),(e',t')$ that share an endpoint, that is $e\cap e'=\{v\}$, we have that $t\neq t'$. A temporal graph that is both simple and proper is called \emph{happy}.

A \emph{temporal path} is a sequence $\langle(e_i, t_i)\rangle$ such that
$\langle e_i \rangle$ is a path in $G$, $\langle t_i \rangle$ is
non-decreasing, and $(e_i,t_i) \in \mathcal{E}$ for all $i$ in the
sequence. If $\langle t_i \rangle$ is increasing, the temporal path is said to be strict. Thus strict temporal paths are a special case of non-strict temporal paths.
A temporal graph \G is (strictly) temporally connected (in class (\textsf{strict}-)\TC) if there exists at least one (strict) temporal path between every ordered pair of vertices. Finally, a temporal spanner of a (\textsf{strict}-)\TC temporal $\G=(V,\mathcal{E},T)$ is a temporal graph $\G'=(V,\mathcal{E}',T')$ such that $\mathcal{E}'\subseteq \mathcal{E}$, $T'\le T$, and $\G' \in (\textsf{strict}$-$)\TC$.

In this paper, we mainly consider the following problem:
\medskip

% A \emph{temporal $(s,z)$-path} (or \emph{temporal path} from $s$ to $z$) of length~$k$ from vertex $s=v_0$ to vertex $z=v_k$ is a sequence $P = \left(\left(v_{i-1},v_i,t_i\right)\right)_{i=1}^k$ of \emph{transitions} such that for all $i\in[k]$ we have that $(\{v_{i-1},v_i\},t_i)\in \mathcal{E}$ and for all $i\in [k-1]$ we have that $t_i \le t_{i+1}$ or $t_i < t_{i+1}$.
% In the latter case, the temporal $(s,z)$-path is called \emph{strict}, in the former case it is \emph{non-strict}. 
% For $s,z\in V$ we say that $s$ is \emph{(strictly) connected} to $z$ in $\mathcal{G}$ if there is a (strict) temporal $(s,z)$-path in $\mathcal{G}$. We say that $\mathcal{G}$ is \emph{(strictly) temporally connected} if for all $u,v\in V$ we have that $u$ is (strictly) connected to $v$.

\problemdef{\textsc{Minimum Temporal Spanner}}{A temporal graph $\mathcal{G}=(V,\mathcal{E},T) \in$ (\textsf{strict}-)$\TC$ and an integer $k$.}{Is there a temporal spanner $\G'=(V,\mathcal{E}',T')$ of $\G$ such that $|\mathcal{E}'|\le k$?}
  % and 
%for all $u,v\in V$ it holds that if $u$ is (strictly) connected to $v$ in $\mathcal{G}$, then $u$ is (strictly) connected to $v$ in 
% $\mathcal{G}'=(V,\mathcal{E}',T)$ is (strictly) temporally connected?

As previously explained, whenever the temporal graph is proper (and a fortiori, if it is happy), the distinction between strict and non-strict temporal paths disappears, which simplifies the definitions. Similarly, whenever the temporal graph is simple, we may refer to edges and time edges interchangeably. This will be the case in most of the paper, although not everywhere.
Finally, if the underlying graph $G_U$ of a temporal graph \G is a tree, we call $\mathcal{G}$ a \emph{temporal tree}. If this tree preserves reachability from a vertex to all others, we call it a temporal out-tree (rooted in that vertex). It is well known that if a temporal graph is \TC, then it contains a temporal out-tree of size $n-1$ rooted in each vertex (see, e.g.~\cite{XFJ03}).

\section{NP-hardness for Happy Temporal Graphs}
\label{sec:hardness}
In this section, we prove that \textsc{Minimum Temporal Spanner} is NP-hard and W[1]-hard when parameterized by the vertex cover number of the footprint, even if the input graph is happy. 
\begin{theorem}
\label{thm:vc-happy-hard}
    \MTS is NP-hard and W[1]-hard when parameterized by the vertex cover number of the
    underlying graph, even on happy temporal graphs.
\end{theorem}

To prove \cref{thm:vc-happy-hard}, we present a parameterized polynomial-time reduction from the \textsc{Multicolored Clique} problem, which is known to be NP-hard and W[1]-hard when parameterized by the number of colors $k$~\cite{fellows2009multipleinterval}.
Given an instance $G=(V_1\uplus\dots\uplus V_k,E)$ of \textsc{Multicolored Clique}, we construct an instance $(\mathcal G=(V,\mathcal E,T),k')$ of \MTS such that 
$\mathcal G$ has a temporal spanner with at most $k'$ edges (we specify $k'$ later) if and only if $G$ has a multicolored clique of size $k$. We assume w.l.o.g.\ that $k\geq5$.

The reduction uses an edge selection gadget for each color pair and consistency gadgets. The consistency gadgets check for each pair of selected edges where the corresponding color pairs share a color, whether the two endpoints of the common color are the same. This then ensures that the selected edges form a multicolored clique. 
Furthermore, an additional gadget (TC gadget) will ensure temporal connectivity between all vertex pairs.

We first provide a very informal description of the reduction to convey the main ideas.
Intuitively, the edge selection gadget for a color pair $\alpha$ simply consists of a pair of vertices $s_\alpha,t_\alpha$ that have vertices representing all edges of the color pair in their common neighborhood, see \cref{fig:vcedgeselection} for an illustration. 
The TC gadget will not provide a temporal path from $s_\alpha$ to~$t_\alpha$.
Hence, to ensure that there is a temporal path from $s_\alpha$ to $t_\alpha$, one vertex in the common neighborhood needs to be ``selected'', and the temporal path from $s_\alpha$ to $t_\alpha$ traverses this vertex. The edges incident with all other vertices in the common neighborhood may be removed when obtaining the temporal spanner.

For each set of two color pairs $\{\alpha, \beta\}$ that share a color $i$, we have a consistency gadget, see \cref{fig:vcconsistency} for an illustration. This gadget needs to provide a temporal path from $s_\alpha$ to $t_\beta$, and from $s_\beta$ to $t_\alpha$, since those temporal connections are not provided by the TC gadget.
Assume that in the two edge selection gadgets, two edges have been selected. Let the endpoint of color $i$ of the selected edge for color pair $\alpha$ be the $\ell$th vertex of color $i$, and let the endpoint of color $i$ of the selected edge for color pair $\beta$ be the $\ell'$th vertex of color $i$. The consistency gadget is constructed in a way that we can only have a temporal path from $s_\alpha$ to $t_\beta$ if $\ell'\ge \ell$. Additionally, we can only have a temporal path from $s_\beta$ to $t_\alpha$ if $\ell\ge \ell'$. Both paths need to exist in a temporal spanner, hence we must have that $\ell=\ell'$, which means that both selected edges have the same endpoint of color $i$.

We first introduce some additional notation that we use to describe the gadgets formally.

\subparagraph{Notation.}
    Fix a linear order $<_i$ on every color class $V_i$ and denote by $r_i(u)$ the rank of $u\in V_i$ in $<_i$.
    Let $\mathcal P:=\{ij:1\leq i<j\leq k\}$ be the set of unordered color pairs and let $\mathcal I:=\{(i,\alpha,\beta):\alpha,\beta\in\mathcal P,\ \alpha\neq\beta,\ i\in\alpha\cap\beta\}$ be the set of triples consisting of a color and any two distinct color pairs containing it.
    For $\alpha=ij\in\mathcal P$, let $E_\alpha:=\{uv\in E:u\in V_i,\ v\in V_j\}$ be the edges between $V_i$ and $V_j$, and, for $e\in E_\alpha$, denote by $\chi_i(e)$ and $\chi_j(e)$ the endpoint of $e$ in $V_i$ and $V_j$, respectively.

    If there exist distinct $\alpha,\beta\in\mathcal P$ and $i\in\alpha\cap\beta$ such that no two edges $e\in E_\alpha$ and $f\in E_\beta$ satisfy $\chi_i(e)=\chi_i(f)$, then the instance is a no-instance. Hence, we may assume that, for every $(i,\alpha,\beta)\in\mathcal I$, there exist edges $e\in E_\alpha$ and $f\in E_\beta$ such that $\chi_i(e)=\chi_i(f)$.

    In the following, we describe the construction of the temporal graph $\mathcal G$. Rather than specifying the time labels explicitly, we partition the temporal edge order into blocks $L_{-\infty}<L_1<L_2<L_3<L_4<L_{+\infty}$. For each edge, we specify its block and, when relevant, its position relative to the other edges in that block. We say that an edge {\it is in a block} if its time label belongs to that block.

\subparagraph{Edge-selection gadgets.}
    For every color pair $\alpha=ij\in\mathcal P$, introduce two vertices $s_\alpha$ and~$t_\alpha$. For every edge $e\in E_\alpha$, introduce an \emph{edge-selector vertex} $z^\alpha_e$ and two edges $s_\alpha z^\alpha_e$ in block $L_1$ and $z^\alpha_e t_\alpha$ in block $L_4$. The edges within each block are ordered arbitrarily. For an illustration, see \cref{fig:vcedgeselection}.

\begin{figure}[t!]
\begin{center}
\begin{tikzpicture}[line width=1pt,yscale=.8]
    \node[vert,fill=gray,label=above:$s_{\alpha}$] (S) at (0,0) {};

    \node[vert,label=above:$z_{e_1}^{\alpha}$] (Z1) at (5,3) {};
    \node[vert,label=above:$z_{e_2}^{\alpha}$] (Z2) at (5,1.5) {};
    \node[vert,label=above:$z_{e_3}^{\alpha}$] (Z3) at (5,0) {};
    \node (z4) at (5,-1) {$\vdots$};
    \node[vert,label=above:$z_{e_n}^{\alpha}$] (Zn) at (5,-3) {};
        
    \node[vert,fill=gray,label=above:$t_{\alpha}$] (T) at (10,0) {};

    \draw (S) --node[timelabel] {$L_1$} (Z1);
    \draw (S) --node[timelabel] {$L_1$} (Z2);
    \draw (S) --node[timelabel] {$L_1$} (Z3);
    \draw (S) --node[timelabel] {$L_1$} (Zn);
    \draw (T) --node[timelabel] {$L_4$} (Z1);
    \draw (T) --node[timelabel] {$L_4$} (Z2);
    \draw (T) --node[timelabel] {$L_4$} (Z3);
    \draw (T) --node[timelabel] {$L_4$} (Zn);

\end{tikzpicture}
    \end{center}
    \caption{Illustration of the edge-selection gadget for color pair $\alpha$. Here, $n=|E_\alpha|$. The gray vertices are part of the vertex cover of the underlying graph.}\label{fig:vcedgeselection}
\end{figure}

    %\begin{figure}[h]
    %    \centering
    %    \includegraphics[width=0.6\textwidth]{figuresketches-mi/image.png}
    %    \caption{Illustration of the edge-selection gadgets between $12$, $23$ and $13$, simulating the choices for $V_1$, $V_2$, and $V_3$.}
    %\end{figure}
    %=====INTUITION
    %The budget will be chosen so that a spanner can retain exactly one such path for every $\alpha$. Retaining the path through $z^\alpha_e$ represents the choice of the original edge $e$ for the color pair $\alpha$. Notice that the reduction chooses an \emph{edge} for every color pair, rather than choosing one vertex independently in every color class. It remains to ensure that all chosen edges incident with a fixed color have a common endpoint.

\subparagraph{Consistency gadgets.}
    For every $(i,\alpha,\beta)\in\mathcal I$, introduce two {\it comparison vertices} $b^{i,-}_{\alpha\to\beta}$ and~$b^{i,+}_{\alpha\to\beta}$. Moreover, for every rank $\ell\in[|V_i|]$, introduce a {\it threshold vertex} $w^{i,\ell}_{\alpha\to\beta}$.
    For every $e\in E_\alpha$, add the edge $z^\alpha_e b^{i,-}_{\alpha\to\beta}$, and, for every $\ell\in[|V_i|]$, add the edge $b^{i,-}_{\alpha\to\beta}w^{i,\ell}_{\alpha\to\beta}$. All these edges are in block $L_2$ in the following order: For each rank $\ell$, first place all edges $z^\alpha_e b^{i,-}_{\alpha\to\beta}$ with $r_i(\chi_i(e))=\ell$, followed by the edge $b^{i,-}_{\alpha\to\beta}w^{i,\ell}_{\alpha\to\beta}$. These pairs of groups are ordered by increasing rank $\ell$, while within each group the edges are ordered arbitrarily.
    Next, for every $\ell\in[|V_i|]$, add the edge $w^{i,\ell}_{\alpha\to\beta}b^{i,+}_{\alpha\to\beta}$, and, for every $f\in E_\beta$, add the edge $b^{i,+}_{\alpha\to\beta}z^\beta_f$. All these edges are in block $L_3$. For each rank $\ell$, first place the edge $w^{i,\ell}_{\alpha\to\beta}b^{i,+}_{\alpha\to\beta}$, followed by all edges $b^{i,+}_{\alpha\to\beta}z^\beta_f$ with $r_i(\chi_i(f))=\ell$. These pairs of groups are again ordered by increasing rank $\ell$, while within each group the edges are ordered arbitrarily.
    This finishes the construction of the consistency gadget for $(i,\alpha,\beta)$. 
    For an illustration, see \cref{fig:vcconsistency}.

\begin{figure}[t!]
\begin{center}
\begin{tikzpicture}[line width=1pt,xscale=.9,yscale=1.3]
    \node[vert,fill=gray,label=above:$s_{\alpha}$] (Sa) at (0,0) {};
    
    \node[vert,label=left:$z_{e_1}^{\alpha}$] (Z1) at (-1.5,-2) {};
    \node[vert,label=left:$z_{e_2}^{\alpha}$] (Z2) at (0,-2) {};
    \node[vert,label=left:$z_{e_3}^{\alpha}$] (Z3) at (1.5,-2) {};
        
    \node[vert,fill=gray,label=below:$t_{\alpha}$] (Ta) at (0,-4) {};

    \draw (Sa) -- (Z1);
    \draw (Sa) -- (Z2);
    \draw (Sa) -- (Z3);
    \draw (Ta) -- (Z1);
    \draw (Ta) -- (Z2);
    \draw (Ta) -- (Z3);

    \node[vert,fill=gray,label=above:$s_{\beta}$] (Sb) at (10,0) {};
    
    \node[vert,label=right:$z_{e_1}^{\beta}$] (Zz1) at (8.5,-2) {};
    \node[vert,label=right:$z_{e_2}^{\beta}$] (Zz2) at (10,-2) {};
    \node[vert,label=right:$z_{e_3}^{\beta}$] (Zz3) at (11.5,-2) {};
        
    \node[vert,fill=gray,label=below:$t_{\beta}$] (Tb) at (10,-4) {};

    \draw (Sb) -- (Zz1);
    \draw (Sb) -- (Zz2);
    \draw (Sb) -- (Zz3);
    \draw (Tb) -- (Zz1);
    \draw (Tb) -- (Zz2);
    \draw (Tb) -- (Zz3);

    \node[vert,fill=gray,label=above:$b^{i,-}_{\alpha\to\beta}$] (Sab) at (3,-.5) {};
    \node[vert,fill=gray,label=above:$b^{i,+}_{\alpha\to\beta}$] (Tab) at (7,-.5) {};
    
    \node[vert,label=above:$w^{i,1}_{\alpha\to\beta}$] (w1) at (5,0) {};
    %\node[vert,label=above:$w^{i,2}_{\alpha\to\beta}$] (w2) at (5,-.5) {};
    \node[vert,label=above:$w^{i,2}_{\alpha\to\beta}$] (w3) at (5,-1) {};

\draw (Sab) --node[timelabel] {$2$} (w1);
%\draw (Sab) --node[timelabel] {$L_2$} (w2);
\draw (Sab) --node[timelabel] {$3$} (w3);
\draw (Tab) --node[timelabel] {$11$} (w1);
%\draw (Tab) --node[timelabel] {$L_3$} (w2);
\draw (Tab) --node[timelabel] {$12$} (w3);

\draw[color=blue] (Z1) edge[bend left=20]  node[timelabel] {$1$} (Sab);
\draw[color=blue] (Z2) edge[bend left=20]  node[timelabel] {$2$} (Sab);
\draw[color=blue] (Z3) edge[bend left=20]  node[timelabel] {$1$} (Sab);
\draw[color=blue] (Zz1) edge[bend left=-20]  node[timelabel] {$12$} (Tab);
\draw[color=blue] (Zz2) edge[bend left=-20]  node[timelabel] {$12$} (Tab);
\draw[color=blue] (Zz3) edge[bend left=-20]  node[timelabel] {$13$} (Tab);

    \node[vert,fill=gray,label=below:$b^{i,+}_{\beta\to\alpha}$] (Tba) at (3,-3.5) {};
    \node[vert,fill=gray,label=below:$b^{i,-}_{\beta\to\alpha}$] (Sba) at (7,-3.5) {};

    \node[vert,label=below:$w^{i,1}_{\beta\to\alpha}$] (ww1) at (5,-3) {};
    %\node[vert,label=below:$w^{i,2}_{\beta\to\alpha}$] (ww2) at (5,-3.5) {};
    \node[vert,label=below:$w^{i,2}_{\beta\to\alpha}$] (ww3) at (5,-4) {};

\draw (Sba) --node[timelabel] {$2$} (ww1);
%\draw (Sba) --node[timelabel] {$L_2$} (ww2);
\draw (Sba) --node[timelabel] {$3$} (ww3);
\draw (Tba) --node[timelabel] {$11$} (ww1);
%\draw (Tba) --node[timelabel] {$L_3$} (ww2);
\draw (Tba) --node[timelabel] {$12$} (ww3);

\draw[color=green!40!black] (Z1) edge[bend left=-20]  node[timelabel] {$12$} (Tba);
\draw[color=green!40!black] (Z2) edge[bend left=-20]  node[timelabel] {$13$} (Tba);
\draw[color=green!40!black] (Z3) edge[bend left=-20]  node[timelabel] {$12$} (Tba);
\draw[color=green!40!black] (Zz1) edge[bend left=20]  node[timelabel] {$1$} (Sba);
\draw[color=green!40!black] (Zz2) edge[bend left=20]  node[timelabel] {$1$} (Sba);
\draw[color=green!40!black] (Zz3) edge[bend left=20]  node[timelabel] {$2$} (Sba);
\end{tikzpicture}
    \end{center}
    \caption{Illustration of the consistency gadget between $\alpha$ and $\beta$ with $\alpha\cap\beta=\{i\}$, ensuring the consistency in the choice of $V_i$. Each direction uses two comparison vertices and one possible threshold vertex for every rank of the common color.
    The labels of the edges illustrate the relative positions of their order blocks within $L_2$ and $L_3$, respectively.
    The gray vertices are part of the vertex cover of the underlying graph.}\label{fig:vcconsistency}
\end{figure}
    
    % \begin{figure}[h]
    %     \centering
    %     \includegraphics[width=0.7\textwidth]{figuresketches-mi/image2.png}
    %     \caption{Illustration of the consistency gadget between $12$, $23$, and $13$, ensuring the consistency in the choice of $V_2$. Each direction uses two comparison vertices and one possible threshold vertex for every rank of the common color.}
    %     \label{fig:vcproof-consistencygadgets}
    % \end{figure}

    Across different consistency gadgets, choose any linear extension of the specified orders and assign distinct labels to incident edges.

Formally, the consistency gadgets provide the following function. 
    Let $e\in E_\alpha$ and $f\in E_\beta$, and write $r_e:=r_i(\chi_i(e))$ and $r_f:=r_i(\chi_i(f))$. For a fixed threshold $\ell$, the path $z^\alpha_e,b^{i,-}_{\alpha\to\beta},w^{i,\ell}_{\alpha\to\beta},b^{i,+}_{\alpha\to\beta},z^\beta_f$ is temporal if and only if $r_e\leq\ell\leq r_f$. Consequently, any such path is temporal if and only if $r_e\leq r_f$.
    Since $\mathcal I$ contains ordered triples, the reverse direction $\beta\to\alpha$ is supplied by the gadget indexed by $(i,\beta,\alpha)$.

Note that the only ``sensitive'' reachabilities are $s_\alpha\rightsquigarrow t_\alpha$ and $s_\alpha\rightsquigarrow t_\beta$ for distinct $\alpha,\beta\in\mathcal P$ with $\alpha\cap\beta\neq\varnothing$.
    These reachabilities must be realizable only via the edge-selection and consistency gadgets.
    In the following, we describe a TC gadget with the purpose to provide all remaining ordered reachabilities without creating an additional path between a sensitive pair.

\subparagraph{TC gadget.}
    Let $\mu:=|\mathcal P|$ and fix an ordering $\mathcal P=\{\alpha_1,\ldots,\alpha_\mu\}$.
    Let $B:=\{b^{i,\sigma}_{\alpha\to\beta}:(i,\alpha,\beta)\in\mathcal I,\ \sigma\in\{-,+\}\}$ be the set of all comparison vertices,
    let $W:=\{w^{i,\ell}_{\alpha\to\beta}:(i,\alpha,\beta)\in\mathcal I,\ \ell\in[|V_i|]\}$ be the set of all threshold vertices,
    let $Z_\alpha:=\{z^\alpha_e:e\in E_\alpha\}$ be the set of all edge-selector vertices of $\alpha\in\mathcal P$, and let $Z:=\bigcup_{\alpha\in\mathcal P}Z_\alpha$.
    Furthermore, define $S:=\{s_\alpha:\alpha\in\mathcal P\}$ and $T:=\{t_\alpha:\alpha\in\mathcal P\}$.

    Let $\mathcal P^2_{\mathrm{disj}}:=\{(\alpha,\beta)\in\mathcal P^2:\alpha\cap\beta=\varnothing\}$. Introduce a {\it connector vertex} $d_{\alpha,\beta}$ for every $(\alpha,\beta)\in\mathcal P^2_{\mathrm{disj}}$, and let $D:=\{d_{\alpha,\beta}:(\alpha,\beta)\in\mathcal P^2_{\mathrm{disj}}\}$.
    For every $\alpha\in\mathcal P$, introduce a vertex $p_\alpha$, and define $P_S:=\{p_\alpha:\alpha\in\mathcal P\}$.
    Introduce {\it hub vertices} $H:=\{a^-,a^+,c^-,c^+\}$.
    Fix an arbitrary terminal $t_\star\in T$.

    We refine $L_{+\infty}$ into ordered subblocks $L_5<L_6<L_7<L_8$, and add the following edges:
    \begin{itemize}
        \item 
        Let $X:=B\cup W\cup Z\cup T\cup D$.
        In block $L_{-\infty}$,
        add the edges
        $c^+t_\star
            <\{xa^-:x\in X\}
            <a^-a^+
            <a^+c^-
            <
            \{a^+t_\alpha:\alpha\in\mathcal P\}
            <\{s_\alpha p_\alpha:\alpha\in\mathcal P\}
            <\{p_\alpha c^-:\alpha\in\mathcal P\}$,
        where sets of edges are ordered arbitrarily.
        Let $p^\star$ be the vertex incident with the last edge in $\{p_\alpha c^-:\alpha\in\mathcal P\}$.

        \item 
        For every $(\alpha,\beta)\in\mathcal P^2_{\mathrm{disj}}$, add the edges $s_\alpha d_{\alpha,\beta}$ and $d_{\alpha,\beta}t_\beta$ in $L_5$ and $L_6$, respectively, ordered arbitrarily.
        For every $\beta\in\mathcal P$, let $d^\star_\beta$ be the connector vertex incident with the last edge in $\{d_{\alpha,\beta}t_\beta:\alpha\cap\beta=\varnothing\}$.
        Define $D^\star:=\{d^\star_\beta:\beta\in\mathcal P\}$.

        \item
        Add the edges $c^-c^+<c^+a^-$ in $L_7$.

        \item
        Define $Y:=B\cup W\cup Z\cup S\cup(D\setminus D^\star)\cup(P_S\setminus\{p^\star\})\cup\{a^+\}$.
        For every $y\in Y$, add the edge $c^+y$ in $L_8$, ordered arbitrarily.
    \end{itemize}
See \cref{fig:tcgadget} for an illustration.

    We refer to $a^-$ and $a^+$ as the \emph{early router} and to $c^-$ and $c^+$ as the \emph{late router}.
    The early router is accessible from $X$ and provides the required reachabilities to $T$.
    Both $X$ and $S$ can enter the late router, but vertices of $S$ reach $c^-$ only after all edges from $a^+$ to $T$ have expired.
    Thus, the late router provides the remaining reachabilities from $S$ without creating an additional reachability from $S$ to $T$.

    We denote by $F_{\mathrm{TC}}$ the set of all edges added in the TC gadget.

\subparagraph{Budget $k'$.}
We set $k'=|F_{\mathrm{TC}}|+2|\mathcal P|+4|\mathcal I|$.
%    Since $|D^\star|=|\mathcal P|=\mu$, we have $|X|=|B|+|W|+|Z|+\mu+|D|$ and $|Y|=|B|+|W|+|Z|+|D|+\mu$.
%    Consequently,
%    $|F_{\mathrm{TC}}|=2|B|+2|W|+2|Z|+4|D|+5\mu+5$.
%    We set $k':=|F_{\mathrm{TC}}|+2|\mathcal P|+4|\mathcal I|$.
%    Equivalently, since $|B|=2|\mathcal I|$, this is
 %   $k'=2|W|+2|Z|+4|D|+8|\mathcal I|+7\mu+5$.
%    Here, $\mu=\binom{k}{2}$, $|\mathcal I|=k(k-1)(k-2)$, and $|D|=\binom{k}{2}\binom{k-2}{2}$.

% Finally, set k = m − 5nc − 4nx, where m = |E| is the number of edges in G, nc is the number of clauses in ϕ and nx is the number of variables in ϕ. 
\medskip
This finishes the construction of the instance $(\mathcal G, k')$ for \MTS, which can clearly be computed in polynomial time from $G$. The graph is clearly simple and proper. We will now examine the properties of $\mathcal G$, starting with the vertex cover number.
\begin{lemma} \label{lem:vcproof-vcbound}
    The footprint of $\mathcal G$ has vertex cover number $O(k^3)$. %\textcolor{ForestGreen}{done}
\end{lemma}
\begin{proof}
    We use $C:=B\cup S\cup T\cup H$ as a vertex cover. Then $V\setminus C=Z\cup W\cup D\cup P_S$ forms an independent set by construction.
    Hence every edge is incident with a vertex of $C$.
    Moreover, $|S|+|T|=2\binom{k}{2}=O(k^2)$, $|B|=2|\mathcal I|=O(k^3)$, and $|H|=4$.
    Thus $|C|=O(k^3)$.
\end{proof}

Next we observe the behavior of the consistency gadgets, which will be used repeatedly in the following connectivity, budget, and soundness arguments.
\begin{lemma}[Consistency gadget paths]
\label{lem:vcproof-consistency-paths}
    Fix $(i,\alpha,\beta)\in\mathcal I$, $e\in E_\alpha$, and $f\in E_\beta$.
    The consistency gadget indexed by $(i,\alpha,\beta)$ contains a temporal path from $z^\alpha_e$ to $z^\beta_f$ if and only if $r_i(\chi_i(e))\leq r_i(\chi_i(f))$.
    It contains no temporal path from $z^\beta_f$ to $z^\alpha_e$.
    Moreover, no temporal path can traverse two consistency gadgets consecutively. %\textcolor{ForestGreen}{done, not 100\% happy with writing}
\end{lemma}
\begin{proof}
    A path through the gadget must use a threshold vertex $w^{i,\ell}_{\alpha\to\beta}$ for some $\ell\in[|V_i|]$.
    Its subpath in $L_2$ is temporal exactly when $r_i(\chi_i(e))\leq\ell$, while its subpath in $L_3$ is temporal exactly when $\ell\leq r_i(\chi_i(f))$.
    Hence such a threshold exists exactly when $r_i(\chi_i(e))\leq r_i(\chi_i(f))$.

    A traversal in the opposite direction would enter the gadget through an edge in $L_3$ and would subsequently have to traverse an edge in $L_2$, which is impossible.
    Similarly, after a completed comparison path exits at a selector vertex in $L_3$, all edges through which that selector could enter the source side of another consistency gadget lie in $L_2$ and are already in the past.
    Entering the target side of another gadget in $L_3$ cannot yield a completed traversal, since the next required edge lies in $L_2$.
    Thus, no temporal path can traverse two consistency gadgets consecutively.
\end{proof}

    \begin{table}[t]
        \centering
        \small
        \renewcommand{\arraystretch}{1.18}
        \newcommand{\grayhline}{\noalign{\color{black!20}\hrule height 0.4pt\color{black}}}
        \begin{tabular}{@{}p{0.08\textwidth}p{0.42\textwidth}p{0.49\textwidth}@{}}
            \hline
            \textbf{Vertices} & \textbf{Reachability role} & \textbf{Role in forcing the TC edges} \\
            \grayhline
            $a^-,a^+$
            & form early router; through it every vertex of $X$ reaches $T$ and $c^-$
            & enforce the edges in $Xa^-$, $a^-a^+$, $a^+c^-$, and $a^+T$ \\
            \grayhline
            $c^-,c^+$
            & form late router; merges the arrivals from $X$ and the delayed arrivals from $S$; provides remaining reachabilities to $Y$ and $a^-$
            & edges $c^-c^+$ and $c^+a^-$ are forced by reachabilities from $S$ to $a^-$; edges $c^+y$ are forced using suitable sources in $S$ \\
            \grayhline
            $p_\alpha$
            & makes edge from $s_\alpha$ to $c^-$ directed 
            & $s_\alpha\rightsquigarrow c^-$ forces both $s_\alpha p_\alpha$ and $p_\alpha c^-$ \\
            \grayhline
            $d_{\alpha,\beta}$
            & realizes $s_\alpha\rightsquigarrow t_\beta$ whenever $\alpha\cap\beta=\varnothing$.
            & both edges incident to $d_{\alpha,\beta}$ are enforced \\
            \grayhline
            $t_\star$
            & provides path $(c^+,t_\star,a^-,a^+)$ by which $c^+$ enters the early router
            & edge $c^+t_\star$ is forced by $c^+\rightsquigarrow t_\star$ \\
            \grayhline
            $p^\star$
            & is reached through the last edge in $\{p_\alpha c^-:\alpha\in\mathcal P\}$; consequently, $c^+p^\star$ is not needed
            & it is the last-element exception needed to force the other edges $c^+p_\alpha$ \\
            \grayhline
            $d^\star_\beta$
            & is reached through the last edge in $\{d_{\alpha,\beta}t_\beta:\alpha\cap\beta=\varnothing\}$; consequently, $c^+d^\star_\beta$ is not needed
            & it is the corresponding last-element exception for the edges $c^+d_{\alpha,\beta}$ \\
            \hline
        \end{tabular}
        \caption{Roles of the auxiliary vertices in the TC gadget.}
        \label{tab:vcproof-tc-vertex-roles}
    \end{table}
Now we show that the TC gadget indeed provides the required reachabilities. See \cref{fig:tcgadget} for an illustration of the TC gadget, and \cref{tab:vcproof-tc-vertex-roles} for a table of the roles of each vertex in the TC gadget. %\todo{MD: the table is mostly for checking  (although I wouldnt mind keeping something similar), but I tend to rather over- than underexplain.}
\begin{lemma} \label{lem:vcproof-tc-completion}
    The graph $\mathcal G$ is temporally connected. %\textcolor{ForestGreen}{done}
\end{lemma}
\begin{proof}
    We first show $S\rightsquigarrow T$.
    The connector vertex $d_{\alpha,\beta}$ gives $s_\alpha\rightsquigarrow t_\beta$ whenever $\alpha\cap\beta=\varnothing$, and the edge-selection gadgets give $s_\alpha\rightsquigarrow t_\alpha$ for every $\alpha\in\mathcal P$. Now consider distinct $\alpha,\beta\in\mathcal P$ with $\alpha\cap\beta=\{i\}$. 
    By our assumption, there exist $e\in E_\alpha$ and $f\in E_\beta$ such that $\chi_i(e)=\chi_i(f)$. Hence $r_i(\chi_i(e))=r_i(\chi_i(f))$, and the consistency gadget indexed by $(i,\alpha,\beta)$ gives a temporal path from $z^\alpha_e$ to $z^\beta_f$ by \cref{lem:vcproof-consistency-paths}. Together with $s_\alpha z^\alpha_e$ and $z^\beta_f t_\beta$, this gives $s_\alpha\rightsquigarrow t_\beta$. 
    We have therefore shown $S\rightsquigarrow T$.

    We next show $X= B\cup W\cup Z\cup T\cup D\rightsquigarrow V$. Every $x\in X$ reaches $a^-$, $a^+$, and $c^-$ through $x,a^-,a^+,c^-$, and it reaches every $t_\beta\in T$ through $x,a^-,a^+,t_\beta$. From $c^-$, it can continue through $c^+$ to every vertex of $Y$. It also reaches $p^\star$ through $c^-,p^\star$ and every $d^\star_\beta\in D^\star$ through $t_\beta,d^\star_\beta$. Since $Y\cup T\cup D^\star\cup\{p^\star,a^-,c^-,c^+\}=V$, this proves $X\rightsquigarrow V$.

    We now show $S\rightsquigarrow V$. 
    Every $s_\alpha\in S$ reaches $c^-$ through $s_\alpha,p_\alpha,c^-$, and from $c^-$ it reaches $c^+$ and then every vertex of $Y$. It reaches $a^-$ through $c^-,c^+,a^-$. Moreover, it reaches $p^\star$ directly if $p_\alpha=p^\star$, and otherwise through $c^-,p^\star$, since $p^\star c^-$ is the last edge in $\{p_\gamma c^-:\gamma\in\mathcal P\}$.
    It remains to reach $D^\star$.
    Fix $d^\star_\beta\in D^\star$. If $\alpha=\beta$ or $\alpha\cap\beta\neq\varnothing$, the path from $s_\alpha$ to $t_\beta$ established above ends before $L_6$ and can be extended by $t_\beta d^\star_\beta$. Suppose instead that $\alpha\cap\beta=\varnothing$.
    If $d_{\alpha,\beta}=d^\star_\beta$, then $s_\alpha$ reaches $d^\star_\beta$ directly. Otherwise, $s_\alpha,d_{\alpha,\beta},t_\beta,d^\star_\beta$ is a temporal path because $d^\star_\beta t_\beta$ is the last connector edge incident with $t_\beta$. 
    Together with $S\rightsquigarrow T$, this proves $S\rightsquigarrow V$. 

    Next, consider $P_S\rightsquigarrow V$.
    Each vertex $p_\alpha\in P_S$ reaches $s_\alpha$ through $p_\alpha s_\alpha$ and reaches $c^-$ directly. It can therefore use the paths from $s_\alpha$ to reach $T\cup D^\star$ and the paths from $c^-$ to reach all remaining vertices. 
    Hence $P_S\rightsquigarrow V$.

    Finally, we show $H=\{a^-,a^+,c^-,c^+\}\rightsquigarrow V$. 
    Each of $a^-$ and $a^+$ reaches $c^-$ and every vertex of $T$. From $c^-$, it reaches $c^+$, every vertex of $Y$, and $p^\star$, and it reaches each $d^\star_\beta$ through $t_\beta$. The same holds for $c^-$, which reaches every vertex of $T$ through $c^-,a^+,t_\beta$. The vertex $c^+$ reaches $a^-$, $a^+$, $c^-$, and every vertex of $T$ through paths beginning with $c^+,t_\star,a^-$. Its edges in $L_8$ reach every vertex of $Y$, while $p^\star$ is reached through $c^-$ and every vertex of $D^\star$ is reached through $T$. 
    Thus $H\rightsquigarrow V$.

    Since $V=X\cup S\cup P_S\cup H$, the conclusions
    $X\rightsquigarrow V$,
    $S\rightsquigarrow V$,
    $P_S\rightsquigarrow V$, and
    $H\rightsquigarrow V$
    show that every vertex reaches every other vertex.
    Therefore, $\mathcal G$ is temporally connected.
\end{proof}

\begin{figure}[t!]
\begin{center}
\begin{tikzpicture}[line width=1pt,yscale=.7]
    \node[vert,fill=gray,label=above:$t^\star$] (TS) at (0,0) {};
    \node[vert,minimum width=30pt,fill=gray,label=above:$T\setminus \{t^\star\}$] (T) at (1,0) {};

    \node[vert,fill=gray,label=above right:$a^-$] (AA) at (.5,-3) {};

    \node[vert,minimum width=30pt,label=above:$D^\star$] (DS) at (-1,-3) {};\node[vert,minimum width=30pt,label=above:$D$] (D) at (-1,-6) {};
    \node[vert,fill=lightgray,minimum width=30pt,label=above right:$B\cup W\cup Z$] (BWZ) at (2,-6) {};

\node[vert,fill=gray,label=above:$a^+$] (AAA) at (6,-3) {};

    \node[vert,fill=gray,label=below left:$c^+$] (CC) at (.5,-9) {};
    \node[vert,fill=gray,label=above right:$c^-$] (CCC) at (6,-9) {};

    \node[vert,minimum width=30pt,fill=gray,label=below:$S$] (S) at (.5,-12) {};

    \node[vert,minimum width=30pt,label=above left:$P_S\setminus\{p^\star\}$] (P) at (5.5,-11.5) {};
    \node[vert,label=below:$p^\star$] (PS) at (6.5,-12.5) {};

    \draw (TS) --node[timelabel] {$2$} (AA);
    \draw (T) --node[timelabel] {$2$} (AA);
    \draw (T) --node[timelabel] {$5$} (AAA);
    \draw (TS) edge[bend left=-25]  node[timelabel] {$5$} (AAA);
    \draw (AA) --node[timelabel] {$3$} (AAA);
    \draw (AA) --node[timelabel] {$2$} (DS);
    \draw (AA) --node[timelabel] {$2$} (D);
    \draw (AA) --node[timelabel] {$2$} (BWZ);
    \draw (TS) edge[bend left=-60]  node[timelabel] {$1$} (CC);
    \draw (AAA) --node[timelabel] {$4$} (CCC);
    \draw (CCC) --node[timelabel] {$5$} (P);
    \draw (CCC) --node[timelabel] {$5$} (PS);
    \draw (S) --node[timelabel] {$4$} (P);
    \draw (S) --node[timelabel] {$4$} (PS);

    \draw[color=blue] (AA) --node[timelabel] {$11$} (CC);
    \draw[color=blue] (D) --node[timelabel] {$12$} (CC);
    \draw[color=blue] (BWZ) --node[timelabel] {$12$} (CC);
    \draw[color=blue] (AAA) edge[bend left=20]  node[timelabel] {$12$} (CC);
    \draw[color=blue] (CC) --node[timelabel] {$10$} (CCC);
    \draw[color=blue] (CC) --node[timelabel] {$12$} (S);
    \draw[color=blue] (CC) edge[bend left=-20]  node[timelabel] {$12$} (P);
\end{tikzpicture}
    \end{center}
    \caption{Illustration of the TC gadget. 
    Large circles visualize vertex sets, small circles visualize single vertices.
    The edges in $L_{-\infty}$ are shown in black and have labels $1$ to $5$, the edges in $L_\infty$ are shown in blue and have labels $10$ to $12$. The labels of the edges illustrate the relative positions of their order blocks within $L_{-\infty}$ and $L_\infty$, respectively. Note that an edge between two vertex sets means that there is an edge between each pair of vertices (where one vertex is from one set and one from the other).
    The gray vertices and vertex sets are part of the vertex cover of the underlying graph. The lightgray vertex set is partially in the vertex cover of the underlying graph ($B$ is part of the vertex cover, but $W$ and $Z$ are not).}\label{fig:tcgadget}
\end{figure}

% \begin{figure}[t]
%     \centering
%     \includegraphics[width=0.45\textwidth]{figuresketches-mi/image3.png}
%     \caption{Rough illustration of the TC gadget. The edges in $L_{-\infty}$ are shown in black and have labels $1$ to $5$, the edges in $L_\infty$ are shown in purple and have labels $10$ to $12$. }
%     \label{fig:vcproof-TCgadget}
% \end{figure}

Next, we show that every edge in the TC gadget needs to be included in every spanner for $\mathcal{G}$. This is crucial for obtaining a lower bound on the size of every minimum temporal spanner.

\begin{lemma}
\label{lem:vcproof-tc-gadget}
    Every edge of $F_{\mathrm{TC}}$ is contained in every temporal spanner of $\mathcal G$. %\textcolor{ForestGreen}{done}
\end{lemma}
\begin{proof}
    We consider the edges of $F_{\mathrm{TC}}$ in order of their blocks.

    First, consider the edges in $L_{-\infty}$.
    For every $x\in X=B\cup W\cup Z\cup T\cup D$ with $x\neq t^\star$, the edge $xa^-$ is required for $x\rightsquigarrow c^-$: without this edge, $x$ cannot enter the path $a^-,a^+,c^-$ before the edges incident with $a^-$ have passed.
    The edge $t^\star a^-$ is required for $t^\star\rightsquigarrow p^\star$: without this edge, $t^\star$ cannot reach $a^-$ and $c^-$ via edges in $L_{-\infty}$, and hence also cannot reach $p^\star$ from $c^-$.
    The edge $a^-a^+$ is required for $a^-\rightsquigarrow c^-$, and the edge $a^+c^-$ is required for $a^+\rightsquigarrow c^-$.
    Similarly, $a^+t_\alpha$ is required for $a^+\rightsquigarrow t_\alpha$ for every $\alpha\in\mathcal P$.
    For every $\alpha\in\mathcal P$, the only temporal path from $s_\alpha$ to $c^-$ is $s_\alpha,p_\alpha,c^-$; hence both $s_\alpha p_\alpha$ and $p_\alpha c^-$ are required.
    Finally, $c^+t_\star$ is required for $c^+\rightsquigarrow t_\star$, since all other edges incident with $c^+$ occur after the edges leading from $a^+$ to $T$.
    Thus, every edge of $F_{\mathrm{TC}}$ in $L_{-\infty}$ is contained in every temporal spanner.

    Next, consider $(\alpha,\beta)\in\mathcal P^2_{\mathrm{disj}}$.
    The edge-selection and consistency gadgets cannot realize $s_\alpha\rightsquigarrow t_\beta$: there is no consistency gadget between disjoint color pairs and, by \cref{lem:vcproof-consistency-paths}, two consistency gadgets cannot be traversed consecutively.
    Moreover, the only edge between $c^+$ and $T$ is $c^+t_\star$, which lies in $L_{-\infty}$ before every path from $S$ can reach $c^+$.
    Consequently, the only temporal path from $s_\alpha$ to $t_\beta$ is
    $s_\alpha,d_{\alpha,\beta},t_\beta$.
    Hence both $s_\alpha d_{\alpha,\beta}$ in $L_5$ and $d_{\alpha,\beta}t_\beta$ in $L_6$ are required.
    Thus, every edge of $F_{\mathrm{TC}}$ in $L_5\cup L_6$ is contained in every temporal spanner.

    We next consider the edges in $L_7$.
    For any $\alpha\in\mathcal P$, every temporal path from $s_\alpha$ to $a^-$ must contain the subpath $c^-,c^+,a^-$.
    Therefore, both $c^-c^+$ and $c^+a^-$ are required.
    Thus, every edge of $F_{\mathrm{TC}}$ in $L_7$ is contained in every temporal spanner.

    It remains to consider the edges $c^+y$ in $L_8$.
    The edge $c^+a^+$ is required for $s_\alpha\rightsquigarrow a^+$.
    Next, fix $p_\alpha\neq p^\star$.
    By the definition of $p^\star$, there is some $\gamma\in\mathcal P$ such that $p_\gamma c^-$ appears after $p_\alpha c^-$.
    Once $s_\gamma$ reaches $c^-$ through $p_\gamma$, the edge $c^-p_\alpha$ is already in the past.
    Hence $c^+p_\alpha$ is required for $s_\gamma\rightsquigarrow p_\alpha$.
    Similarly, fix $d_{\alpha,\beta}\neq d^\star_\beta$.
    By the definition of $d^\star_\beta$, there is some $\gamma$ with $\gamma\cap\beta=\varnothing$ such that $d_{\gamma,\beta}t_\beta$ appears after $d_{\alpha,\beta}t_\beta$.
    Once $s_\gamma$ reaches $t_\beta$ through $d_{\gamma,\beta}$, the edge $t_\beta d_{\alpha,\beta}$ is already in the past.
    Hence $c^+d_{\alpha,\beta}$ is required for $s_\gamma\rightsquigarrow d_{\alpha,\beta}$.
    Now let $y\in Z\cup B\cup W$.
    The colors occurring in the index of $y$ form a set of size at most three.
    Since $k\geq5$, we can choose $\gamma\in\mathcal P$ disjoint from these colors.
    The edge-selection and consistency gadgets cannot provide $s_\gamma\rightsquigarrow y$.
    Therefore, every temporal path from $s_\gamma$ to $y$ uses $c^+y$, and hence $c^+y$ is required.
    Finally, let $y=s_\alpha\in S$ and choose any $\gamma\neq\alpha$.
    The edge-selection and consistency gadgets do not provide $s_\gamma\rightsquigarrow s_\alpha$.
    Consequently, $c^+s_\alpha$ is required for this reachability.

    We have shown that all edges of $F_{\mathrm{TC}}$ in
    $L_{-\infty}$, $L_5$, $L_6$, $L_7$, and $L_8$ are contained in every temporal spanner.
\end{proof}

\cref{lem:vcproof-tc-gadget} allows us to show that every minimum temporal spanner needs to contain at least $k'$ edges. Furthermore, we can also show that minimum spanners of size exactly $k'$ have properties that, intuitively, imply that the edge selection gadgets and consistency gadgets are used in the intended way.

\begin{lemma}[Tight budget]
\label{lem:vcproof-tight-budget}
    Every temporal spanner of $\mathcal G$ contains at least $k'$ edges.
    Moreover, every temporal spanner with at most $k'$ edges contains:  %\textcolor{ForestGreen}{done}
    \begin{enumerate}
        \item every edge of $F_{\mathrm{TC}}$;
        \item for every $\alpha\in\mathcal P$, exactly two edge-selection edges, namely $s_\alpha z^\alpha_{e_\alpha}$ and $z^\alpha_{e_\alpha}t_\alpha$ for some $e_\alpha\in E_\alpha$; and
        \item for every $(i,\alpha,\beta)\in\mathcal I$, exactly four edges forming a path $z^\alpha_{e_\alpha},b^{i,-}_{\alpha\to\beta},w^{i,\ell}_{\alpha\to\beta},b^{i,+}_{\alpha\to\beta},z^\beta_{e_\beta}$ for some $\ell\in[|V_i|]$.
    \end{enumerate}
\end{lemma}
\begin{proof}
    By \cref{lem:vcproof-tc-gadget}, every temporal spanner contains all $|F_{\mathrm{TC}}|$ edges of the TC gadget.

    Fix $\alpha\in\mathcal P$.
    The spanner must realize $s_\alpha\rightsquigarrow t_\alpha$.
    By \cref{lem:vcproof-consistency-paths}, a path cannot leave $Z_\alpha$ through a consistency gadget and return to $Z_\alpha$.
    Hence every such path is of the form $s_\alpha,z^\alpha_e,t_\alpha$ for some $e\in E_\alpha$.
    Consequently, the spanner contains at least two edge-selection edges for every $\alpha$, and hence at least $2|\mathcal P|$ edge-selection edges in total.

    Next, fix $(i,\alpha,\beta)\in\mathcal I$.
    The spanner must realize $s_\alpha\rightsquigarrow t_\beta$.
    By \cref{lem:vcproof-consistency-paths}, every such path uses the consistency gadget indexed by $(i,\alpha,\beta)$ in its forward direction and contains four edges of this gadget.
    Since the consistency gadgets are edge-disjoint, this gives at least $4|\mathcal I|$ consistency edges.

    The three classes of edges are pairwise disjoint.
    Thus, every temporal spanner contains at least $|F_{\mathrm{TC}}|+2|\mathcal P|+4|\mathcal I|=k'$ edges.

    Suppose now that the spanner contains at most $k'$ edges.
    Equality must hold in each of the preceding lower bounds.
    In particular, it contains exactly two edge-selection edges for every $\alpha$.
    Since these edges must realize $s_\alpha\rightsquigarrow t_\alpha$, they are $s_\alpha z^\alpha_{e_\alpha}$ and $z^\alpha_{e_\alpha}t_\alpha$ for one common edge $e_\alpha\in E_\alpha$.
    Similarly, the spanner contains exactly four edges from every consistency gadget.
    The unique retained $s_\alpha$-edge and $t_\beta$-edge, together with \cref{lem:vcproof-consistency-paths}, imply that these edges form a path $z^\alpha_{e_\alpha},b^{i,-}_{\alpha\to\beta},w^{i,\ell}_{\alpha\to\beta},b^{i,+}_{\alpha\to\beta},z^\beta_{e_\beta}$ for some threshold $\ell\in[|V_i|]$.
\end{proof}

Now we show the correctness of the reduction. We first show that if $G$ contains a multicolored clique, we can construct a temporal spanner of the desired size.

\begin{lemma}[Completeness]
\label{lem:completeness-sketch}
    If $G$ contains a multicolored clique of size $k$, then $\mathcal G$ has a
    temporal spanner with at most $k'$ edges. %\textcolor{ForestGreen}{done, check again}
\end{lemma}
\begin{proof}
    Let $v_i\in V_i$, $i\in[k]$, form a multicolored clique. For every
    $\alpha=ij\in\mathcal P$, let $e_\alpha:=v_iv_j$ and pick the two edges of the selection path $s_\alpha z^\alpha_{e_\alpha}t_\alpha$ into the spanner. For every $(c,\alpha,\beta)\in\mathcal I$, let $\ell_c:=r_c(v_c)$ and pick the four edges of the consistency path $z^\alpha_{e_\alpha},b^{c,-}_{\alpha\to\beta},w^{c,\ell_c}_{\alpha\to\beta},b^{c,+}_{\alpha\to\beta},z^\beta_{e_\beta}$ into the spanner.
    Both selected edges have endpoint $v_c$ in color class $V_c$, so their ranks are equal to $\ell_c$ and this path is temporal.
    Finally, pick all edges of $F_{\mathrm{TC}}$ into the spanner.
    The graph containing only the edges of the spanner has exactly $|F_{\mathrm{TC}}|+2|\mathcal P|+4|\mathcal I|=k'$ edges and realizes every sensitive reachability ($S\rightsquigarrow T$) and contains all edges of $F_{\mathrm{TC}}$; hence the reachability argument from the proof of \cref{lem:vcproof-tc-completion} shows that it is temporally connected.
    Consequently, $\mathcal G$ has a temporal spanner with at most $k'$ edges.
\end{proof}

For the opposite direction,  \cref{lem:vcproof-tight-budget} gives us enough insight into the structure of a temporal spanner of size $k'$, that we can find a multicolored clique in $G$.

\begin{lemma}[Soundness]
\label{lem:soundness-sketch}
    If $\mathcal G$ has a temporal spanner with at most $k'$ edges, then $G$ contains a multicolored clique of size $k$. %\textcolor{ForestGreen}{done, check again}
\end{lemma}
\begin{proof}
    By \cref{lem:vcproof-tight-budget}, for every $\alpha\in\mathcal P$ the spanner selects one edge $e_\alpha\in E_\alpha$. Fix a color~$c$. For any two distinct $\alpha,\beta\in\mathcal P$ with $c\in\alpha\cap\beta$, \cref{lem:vcproof-tight-budget} forces the consistency paths indexed by $(c,\alpha,\beta)$ and $(c,\beta,\alpha)$ to use the selected vertices $z^\alpha_{e_\alpha}$ and $z^\beta_{e_\beta}$. By \cref{lem:vcproof-consistency-paths}, their temporal feasibility implies both $r_c(\chi_c(e_\alpha))\leq r_c(\chi_c(e_\beta))$ and $r_c(\chi_c(e_\beta))\leq r_c(\chi_c(e_\alpha))$. Thus all selected edges incident with color $c$ have the same endpoint; denote it by $v_c\in V_c$. For every $i<j$, the selected object $e_{ij}\in E_{ij}$ therefore has endpoints $v_i$ and $v_j$.
    Consequently, $v_iv_j\in E(G)$ for every $i<j$, and $\{v_1,\dots,v_k\}$ is a multicolored clique.
\end{proof}

\cref{lem:completeness-sketch,lem:soundness-sketch} together with \cref{lem:vcproof-vcbound,lem:vcproof-tc-completion} and the fact that our reduction is computable in polynomial time, gives us \cref{thm:vc-happy-hard}.

\section{XP-Algorithm for Vertex Cover Number}
\label{sec:xp-vc}
In this section, we prove that \textsc{Minimum Temporal Spanner} on happy temporal graphs admits an XP-algorithm when parameterized by the vertex cover number of the underlying graph. 
%We present the algorithm for the setting of happy temporal graph, but at the end we give an argument, that the techniques (and hence the result) generalizes to the general setting.
To the best of our knowledge, this is the first non-trivial algorithm for \textsc{Minimum Temporal Spanner} that runs on temporal graphs whose underlying graph is not a tree.
\cref{thm:vc-happy-hard} shows that we presumably cannot improve the algorithm to have an FPT running time.

\begin{theorem}\label{thm:vcn}
    \textsc{Minimum Temporal Spanner} is in XP when parameterized by the vertex cover number of the underlying graph if the input graph is happy.
\end{theorem}

A key ingredient to prove \cref{thm:vcn} is a lemma that shows that any minimum temporal spanner consists of a small number of temporal out-trees, and at most one additional edge per vertex. This insight on the structure of temporal spanners allows us to design a ``guess-and-check'' algorithm and it may be of independent interest. For simplicity, we formulate the lemma in the context of happy temporal graphs, although it should hold more generally.

\begin{lemma}[VC-Tree Lemma]
  \label{lem:k-trees}
  Let $\G$ be a happy \TC graph, let $X$ be a vertex cover of the underlying graph $G_U$ with $|X|=d$. Let $\mathcal{S}$ be a minimum spanner of $\G$. Then, $\mathcal{S}$ is a union of at most~$d$ temporal out-trees rooted in $X$ and at most one extra edge incident to each vertex in $V\setminus X$.
\end{lemma}
\begin{proof}
Let $\G$ be a happy \TC graph, let $X=\{x_1,x_2,\ldots,x_d\}$ be a vertex cover of the underlying graph $G_U$ with $|X|=d$. Let $\mathcal{S} \subseteq \G$ be a minimum spanner of $\G$.
  As $\mathcal{S} \in \TC$, we have that for every vertex $v\in V$, there exists a temporal out-tree $\mathcal{T}_v \subseteq \mathcal{S}$ that allows $v$ to reach all other vertices. For every $v\in V\setminus X$, the minimum edge of $\mathcal{T}_v$ is an edge $\{v,x\}$ such that $x \in X$. We can thus partition $V\setminus X$ into $d$ sets $V_1,V_2,\cdots,V_d$ such that $v\in V_i$ if and only if the minimum edge of $\mathcal{T}_v$ is $\{v,x_i\}$. Now, let 
  \[
  s_i=\argmax_{v\in V_i} \{(\{v,x_i\},t)\},
  \] 
  that is, $s_i$ is the vertex in $V_i$ whose edge to $x_i$ is maximum. Then, for every $V_i$, all the vertices $v$ in $V_i\setminus \{s_i\}$ can reach all the vertices of $V$ using the edge $\{v,x_i\}$ followed by the edges of $\mathcal{T}_{s_i}$ (same for $s_i$ and $x_i$, using $\mathcal{T}_{s_i}$ directly). Note that the tree $\mathcal{T}_{s_i}$ can indistinctly be seen as rooted in $s_i$ or in~$x_i$. Clearly, the subset of $\mathcal{S}$ characterized in this way achieves temporal connectivity, and as $\mathcal{S}$ is minimum, it is $\mathcal{S}$ itself.
\end{proof}

%\begin{proof}[Proof of \cref{thm:vcn}]
Informally speaking, given an instance of \textsc{Minimum Temporal Spanner} $(\mathcal{G},k)$, the algorithm consists of the following steps:
\begin{enumerate}
    \item We compute a minimum vertex cover $X$ for the underlying graph $G_U$. 
    \item We guess the temporal out-tree $\mathcal{T}_x$ for each $x\in X$.\footnote{By ``guess'' we mean that we enumerate and try out sufficiently many possibilities (not necessarily all possibilities).}
    \item We compute an optimal candidate for up to one extra edge incident with $v$ for each $v\in V\setminus X$.
    \item We define the spanner $\mathcal{S}$ as the union of all guessed out-trees and additional edges.
    \item We verify that $\mathcal{S}$ is temporally connected and check whether its size is at most $k$.
\end{enumerate}

The second and third steps are the main part of the algorithm. Since a temporal out-tree $\mathcal{T}_x$ contains $n-1$ time edges, we cannot naively guess the subset of time edges that forms the out-tree. To circumvent this, we, informally speaking, first guess the ``internal structure'' of each out-tree. In a second step, we compute an optimal way to attach the vertices in $V\setminus X$ as leaves to the out-trees in which they are not already contained as internal vertices, and how to pick up to one extra edge incident with $v$ for each $v\in V\setminus X$.

To formalize this, we introduce the notion of \emph{templates} and several additional notions related to templates.

\begin{definition}[Template]
Given a set $X$, a \emph{template} $T$ is a (static) directed out-tree $T=(U,A)$ with $X\subseteq U$, such that:
\begin{itemize}
\item The root of $T$ is some vertex $x^\star\in X$.
\item For each $v\in U\setminus X$, we have $N^+(v)\subseteq X$ \hfill($N^+(v)$ denotes the out-neighbors of~$v$).
\item All leaves of $T$ are vertices from $X$.
\end{itemize}
\end{definition}

From now on, assume we are given a happy \TC $\mathcal{G}$ with underlying graph $G_U$, and let $X$ denote a minimum vertex cover of $G_U$.
The set of nodes $U\setminus X$ can be seen as placeholders to be instantiated with some vertices from $V\setminus X$. Note that the templates do not contain leaves that are not in $X$. In particular, this means that the size of a template is in $O(|X|)$.

Let $\mathcal{T}_x$ be a temporal out-tree rooted at some vertex $x\in X$ in $\mathcal{G}$ and let $T=(U,A)$ be a template. We say that $\mathcal{T}_x$ is \emph{compatible} with $T$ if there is a mapping $\zeta:U\setminus X\rightarrow V\setminus X$ from $U\setminus X$ to vertices in $V\setminus X$ such that the following holds: (See \cref{fig:simpletemplate} for an illustration.)
\begin{itemize}
    %\item If $v^\star\in U$ is the root of $T$, then $\zeta(v^\star)=v$.
    \item If $(x_1,x_2)\in A$ with $x_1,x_2\in X$, then there is a temporal path in $\mathcal{T}_x$ from $x$ to a leaf vertex of $\mathcal{T}_x$ that visits $x_1$ and $x_2$ consecutively in this order.
    \item If $(u,x')\in A$ with $x'\in X$ and $u\in U\setminus X$, then there is a temporal path in $\mathcal{T}_x$ from $x$ to a leaf vertex of~$\mathcal{T}_x$ that visits $\zeta(u)$ and $x'$ consecutively in this order.
    \item If $(x',u)\in A$ with $x'\in X$ and $u\in U\setminus X$, then there is a temporal path in $\mathcal{T}_x$ from~$x$ to a leaf vertex of $\mathcal{T}_x$ that visits $x'$ and $\zeta(u)$ consecutively in this order.
\end{itemize}
Given a template $T=(U,A)$ and a mapping $\zeta:U\setminus X\rightarrow V\setminus X$, we say that $(T,\zeta)$ is \emph{eligible} if there is a temporal out-tree $\mathcal{T}'_x$ rooted at some vertex $x\in X$ in $\mathcal{G}[X\cup \zeta(U\setminus X)]$ such that $\mathcal{T}'_x$ is compatible with $T$. Note that in particular, every template for which there exists an out-tree that is compatible with it, has a mapping with which it is eligible.
Furthermore, given an eligible pair $(T,\zeta)$, the out-tree $\mathcal{T}'_x$ that confirms its eligibility is uniquely defined. We call this out-tree the \emph{eligibility witness} of $(T,\zeta)$.

\begin{figure}[t]
\begin{center}
\begin{tikzpicture}[line width=1pt,scale=.9,xscale=.9]

    \node[vert,fill=gray,label=below:$x^\star$] (v) at (-3,0) {};
    \node[vert2,label=below:$u_1$] (x1) at (-1,1) {};
    \node[vert,fill=gray,label=below:$x_2$] (x2) at (-1,-1) {};
    \node[vert,fill=gray,label=below:$x_1$] (x3) at (1,1) {};
    \node[vert,fill=gray,label=below:$x_3$] (x4) at (3,0) {};
    \node[vert2,label=below:$u_2$] (u1) at (1,-.5) {};
    \node[vert2,label=below:$u_3$] (u2) at (1,-1.5) {};
    \node[vert,fill=gray,label=below:$x_4$] (x5) at (3,-2) {};
      
    \draw[diredge] (v) -- (x1);
    \draw[diredge] (x1) -- (x3);
    \draw[diredge] (v) -- (x2);  
    \draw[diredge] (x2) -- (u1); 
    \draw[diredge] (x2) -- (u2);      
    \draw[diredge] (u1) -- (x4); 
    \draw[diredge] (u2) -- (x5);
\end{tikzpicture}
\begin{tikzpicture}[line width=1pt,scale=.9,xscale=.9]
\node (y) at (-5,0) {};

    \node[vert,fill=gray,label=below:$x^\star$] (v) at (-3,0) {};
    \node[vert,label=below:$\zeta(u_1)$] (x1) at (-1,1) {};
    %\node[vert] (u1) at (-1.5,2) {};
    %\node[vert] (u2) at (-.5,2) {};
    \node[vert,fill=gray,label=above:$x_2$] (x2) at (-1,-1) {};
    \node[vert] (u3) at (-1.5,-2) {};
    \node[vert] (u4) at (-.5,-2) {};
    \node[vert,fill=gray,label=below:$x_1$] (x3) at (1,1) {};
    \node[vert] (u5) at (2,1.5) {};
    \node[vert] (u6) at (2,.5) {};
    \node[vert,fill=gray,label=below:$x_3$] (x4) at (3,0) {};
    \node[vert] (u7) at (4,.75) {};
    \node[vert] (u8) at (4,0) {};
 
    \node[vert] (u9) at (4,-.75) {};
    \node[vert,label=below:$\zeta(u_2)$] (u2) at (1,-.5) {};
    \node[vert,label=below:$\zeta(u_3)$] (u) at (1,-1.5) {};
    \node[vert,fill=gray,label=below:$x_4$] (x5) at (3,-2) {};
    \node[vert] (u10) at (4,-1.5) {};
    \node[vert] (u11) at (4,-2.5) {};
      
    \draw (v) --node[timelabel] {$1$} (x1);
    %\draw (x1) --node[timelabel] {$2$} (u1);
    %\draw (x1) --node[timelabel] {$4$} (u2);
    \draw (x1) --node[timelabel] {$3$} (x3);
    \draw (x3) --node[timelabel] {$4$} (u5);
    \draw (x3) --node[timelabel] {$5$} (u6);
    \draw (v) --node[timelabel] {$4$} (x2);  
    \draw (x2) --node[timelabel] {$6$} (u); 
    \draw (x2) --node[timelabel] {$5$} (u3);
    \draw (x2) --node[timelabel] {$6$} (u4);     
    \draw (x2) --node[timelabel] {$7$} (u2);     
    \draw (u2) --node[timelabel] {$8$} (x4);     
    \draw (x4) --node[timelabel] {$8$} (u7);    
    \draw (x4) --node[timelabel] {$9$} (u8);    
    \draw (x4) --node[timelabel] {$9$} (u9); 
    \draw (u) --node[timelabel] {$7$} (x5);
    \draw (x5) --node[timelabel] {$8$} (u10);
    \draw (x5) --node[timelabel] {$9$} (u11);
\end{tikzpicture}
    \end{center}
    \caption{Illustration of a template on the left and a temporal tree that is compatible with it on the right. Gray vertices are the ones in $X$.}\label{fig:simpletemplate}
\end{figure}

We now have the ingredients ready to describe the algorithm to prove \cref{thm:vcn}.

\begin{algorithm}[Computing a Minimum Spanner]\label{alg:algorithm}~\\
\textbf{Input:} A happy \TC temporal graph $\mathcal{G}$, and an integer $k$.
\begin{enumerate}
    \item Compute a minimum vertex cover $X$ for the underlying graph $G_U$.
    \item For each $x\in X$, guess a template $T_x=(U_x,A_x)$ and a mapping $\zeta_x: U_x\setminus X\rightarrow V\setminus X$ for the template $T_x$ such that $(T_x,\zeta_x)$ is eligible.
    \item Given a guess $(T_x,\zeta_x)_{x\in X}$, denote with $N_x(v)$ the set of neighbors of $v$ in the underlying graph of the eligibility witness $\mathcal{T}_x'$ of $(T_x,\zeta_x)$. Let $N(v)$ denote the set of neighbors of $v$ in $G_U$.
    Iterate over all $v\in V\setminus X$ and do the following:

    Compute the smallest subset $W_v\subseteq N(v)\setminus N_x(v)$ such that the following holds.
    \begin{itemize}
        \item There is an $x\in W_v\cup N_x(v)$ such that $(\{v,x\},t)$ is the time edge with the smallest label in the eligibility witness $\mathcal{T}_x'$ of $(T_x,\zeta_x)$, or
        $(\{v,x\},t)\in\mathcal{E}$ and $t$ is smaller than each label of the time edges in $\mathcal{T}_x'$.
        \item For each $x\in X$ there is an $x'\in W_v\cup N_x(v)$ such that $v$ is reachable from $x$ in $\mathcal{T}_x'+(\{v,x'\},t)$, where $\mathcal{T}_x'$ is the eligibility witness of $(T_x,\zeta_x)$ and $(\{v,x'\},t)\in \mathcal{E}$.
    \end{itemize}
    \item Let $\mathcal{E}^\star=\{(\{v,x\},t)\mid (\{v,x\},t)\in \mathcal{E}\wedge x\in W_v\}$ and let $\mathcal{T}'_x=(V'_x,\mathcal{E}_x,T_x)$. Compute
    \[
    \mathcal{S}=(V,\bigcup_{x\in X}\mathcal{E}_x\cup \mathcal{E}^\star,T^\star),
    \]
    where $T^\star$ is the maximum label among all time edges in $\mathcal{S}$. Check whether $\mathcal{S}$ is \TC and whether $\mathcal{S}$ has at most $k$ time edges.

    If the checks succeed, output YES. Otherwise, reject the current guess. If all guesses have been rejected, output NO.
\end{enumerate}
\end{algorithm}

We first analyze the running time of \cref{alg:algorithm}.

\begin{lemma}\label{lem:algoruntime}
    \cref{alg:algorithm} runs in time $n^{O(d^2)}$.
\end{lemma}
\begin{proof}
    In the first step of \cref{alg:algorithm}, we compute a minimum vertex cover of $G_U$. It is well-known, that this can be done in $2^d\cdot n^{O(1)}$ time~\cite{chen2006improved}.

    Consider the second step of \cref{alg:algorithm}. By definition, the size of a template is in $O(d)$. We can upper-bound the number of different templates by $d^{O(d)}$ using Cayley's formula (the number of labeled trees on $\ell$ vertices is $\ell^{\ell-2}$). Furthermore, we can upper-bound the number of different mappings by $n^{O(d)}$. Hence, the overall number of template-mapping-pairs $(T_x,\zeta_x)$ for some $x\in X$ is in $n^{O(d)}$. Checking whether $(T_x,\zeta_x)$ is eligible can be done straightforwardly in polynomial time~\cite{KKK02}.
    The overall number of possibilities for lists $(T_x,\zeta_x)_{x\in X}$ of template-mapping-pairs is in $n^{O(d^2)}$.

    Consider the third step of \cref{alg:algorithm}. Assume that $(T_x,\zeta_x)_{x\in X}$ is fixed. Then, for every vertex $v\in V$, all subsets of $N(v)$ need to be checked. The checks in the algorithm can straightforwardly be done in polynomial time~\cite{KKK02}. Hence, for $(T_x,\zeta_x)_{x\in X}$, the third step takes time $2^{O(d)}\cdot n^{O(1)}$.

    The fourth step of \cref{alg:algorithm} can straightforwardly be done in polynomial time~\cite{KKK02}. It follows that the overall running time of \cref{alg:algorithm} is in $n^{O(d^2)}$. 
\end{proof}

Next, we show that \cref{alg:algorithm} is correct.

\begin{lemma}\label{lem:algocorrect}
    \cref{alg:algorithm} is correct.
\end{lemma}

\begin{proof}
First, if the algorithm outputs yes, then the input $(\mathcal{G},k)$ is clearly a yes-instance of \textsc{Minimum Temporal Spanner}, since the algorithm checks in the fourth step whether $\mathcal{S}$ is \TC\ and whether its size is at most $k$.

Now for the opposite direction, assume that $(\G,k)$ is a yes-instance of \textsc{Minimum Temporal Spanner}. Then $\G$ admits some temporal spanner $\mathcal{S}=(V,\mathcal{E}_\mathcal{S},T_\mathcal{S})$ of size at most $k$. Assume w.l.o.g.\ that $\mathcal{S}$ is minimum. Then by \cref{lem:k-trees}, we have that $\mathcal{S}$ is a union of at most $d$ simple temporal out-trees $\mathcal{T}_x$ with $x\in X$ and, for each vertex $v\in V\setminus X$, at most one additional edge incident with $v$.

For each temporal out-tree $\mathcal{T}_x$ with $x\in X$ there exists, by definition, a template $T_x$ such that $\mathcal{T}_x$ is compatible with $T_x$ and there exists a mapping $\zeta_x$ that witnesses compatibility. 
%Furthermore, we have that all vertices in $V\setminus X$ that are not in the image of $\zeta$ must be leaves in $\mathcal{T}_x$. 
Then the enumeration in Step 2 of the algorithm will find such $(T_x,\zeta_x)$ for all $x\in X$.

For each $x\in X$, let $\mathcal{T}'_x=(V'_x,\mathcal{E}_x,T_x)$ denote the eligibility witness of $(T_x\zeta_x)$.
Denote $\mathcal{E}'=\bigcup_{x\in X} \mathcal{E}_x$ and note that by definition, $\mathcal{E}'\subseteq \mathcal{E}_\mathcal{S}$. Denote $\mathcal{E}^{\star\star}=\mathcal{E}_\mathcal{S}\setminus \mathcal{E}'$. 
For each $v\in V\setminus X$, let $W_v$ denote teh set computed in Step 3 of \cref{alg:algorithm}.
Let $\mathcal{E}^\star=\{(\{v,x\},t)\mid (\{v,x\},t)\in \mathcal{E}\wedge x\in W_v\}$ and let $\mathcal{T}'_x=(V'_x,\mathcal{E}_x,T_x)$, as defined in Step 4 of \cref{alg:algorithm}.
In the remainder, we prove the following.
\begin{enumerate}
    \item $\mathcal{S}'=(V,\bigcup_{x\in X}\mathcal{E}_x\cup \mathcal{E}^\star,T^\star)$ is \TC, and
    \item $|\mathcal{E}^\star|\le |\mathcal{E}^{\star\star}|$.
\end{enumerate}
Note that the above two statements imply that \cref{alg:algorithm} outputs YES in Step 4.

For the first statement, assume for contradiction that $\mathcal{S}'$ is not \TC. Then there are $u,v\in V$ such that $u\neq v$ and there is no temporal path from $u$ to $v$ in $\mathcal{S}'$.
Since $\mathcal{S}$ is \TC, there is a temporal path from $u$ to $v$ in $\mathcal{S}$. 
Consider the case that $u\in X$. Then the temporal path from $u$ to $v$ is entirely contained in $\mathcal{T}_u$. If $v$ is a vertex of $\mathcal{T}'_u$, then the path is also entirely contained in $\mathcal{T}'_u$ and hence also in $\mathcal{S}'$. Hence, we may assume that $v$ is not contained in $\mathcal{T}'_u$, and in particular, then we have that $v\in V\setminus X$. Consider the set $W_v$ computed in the Step 3 of \cref{alg:algorithm}. By definition, there is an $x\in W_v$ such that $v$ is reachable from $u$ in $\mathcal{T}_u'+(\{v,x\},t)$ and $(\{v,x\},t)\in\mathcal{E}$. It follows that $(\{v,x\},t)\in\mathcal{E}^\star$ and hence, that the temporal path from $u$ to $v$ is contained in $\mathcal{S}'$, a contradiction.

Now consider the case that $u\in V\setminus X$ and consider the set $W_u$ computed in Step 3 of \cref{alg:algorithm}. By definition, there is an $x\in W_v\cup N_x(v)$ such that $(\{v,x\},t)$ is the time edge with the smallest label in the eligibility witness $\mathcal{T}_x'$ of $(T_x,\zeta_x)$, or $(\{v,x\},t)\in\mathcal{E}$ and $t$ is smaller than each label of the time edges in $\mathcal{T}_x'$.
If $v$ is contained in $\mathcal{T}_x'$, then $v$ is reachable from $u$ in $\mathcal{T}_u'+(\{u,x\},t)$. Since $(\{v,x\},t)\in\mathcal{E}^\star$, we have that $v$ is also reachable from $u$ in $\mathcal{S}'$, a contradiction. Hence, we may assume that $v$ is not contained in $\mathcal{T}'_x$, and in particular, then we have that $v\in V\setminus X$. Consider the set $W_v$ computed in the Step 3 of \cref{alg:algorithm}. By definition, there is an $x'\in W_v$ such that $v$ is reachable from $x$ in $\mathcal{T}_x'+(\{v,x'\},t')$ and $(\{v,x'\},t')\in\mathcal{E}$. Furthermore, we have that $v$ is reachable from $u$ in  $\mathcal{T}_x'+(\{u,x\},t)+(\{v,x'\},t')$.
Since $(\{v,x'\},t')\in\mathcal{E}^\star$, we have that $v$ is reachable from $u$ in $\mathcal{S}'$, a contradiction. We can conclude that $\mathcal{S}'$ is \TC.

For the second statement, consider the following. 
We define set $W'_v$ for each $v\in V\setminus X$ as follows.
%For any $u,v\in V\setminus X$ mit $u\neq v$, the computation of $W_v$ in Step 3 of \cref{alg:algorithm}is independent from the computation of $W_u$. Now consider some $v\in V\setminus X$. 
\cref{lem:k-trees} and its proof imply that there is an out-tree $\mathcal{T}_v$ rooted at $v$ in $\mathcal{S}$ that equals out-tree $\mathcal{T}_x$ for some $x\in X$ up to the edges incident with $v$. 
In particular, this means that edge $(\{v,x\},t)$ is the edge with the smallest label in $\mathcal{T}_x$, or that edge $(\{v,x\},t)$ has a label that is smaller than each label in $\mathcal{T}_x$.
In the latter case, we add $x$ to $W'_v$. 

Furthermore, for each $x\in X$, we have that $v$ is reachable from $x$ in $\mathcal{T}_x$. Let $x'\in X$ be the vertex visited by the temporal path from $x$ to $v$ in $\mathcal{T}_x$ right before $v$. If $v$ is a leaf in $\mathcal{T}_x$, then we add $x'$ to $W'_v$.

Note that for each $v\in V\setminus X$, the spanner $\mathcal{S}$ contains $|W'_v|$ edges incident with $v$ that are in $\mathcal{E}^{\star\star}$. This follows from the definition of $\mathcal{E}^{\star\star}$. Hence, we have that $|\mathcal{E}^{\star\star}|=\sum_{v\in V\setminus X} |W'_v|$.
Furthermore, for each $v\in V\setminus X$, the set $W'_v$ fulfills the conditions of Step 3 of \cref{alg:algorithm}. Hence, we have that $|W_v|\le |W'_v|$. Furthermore, by definition we also have that $|\mathcal{E}^{\star}|=\sum_{v\in V\setminus X} |W_v|$. We can conclude that $|\mathcal{E}^\star|\le |\mathcal{E}^{\star\star}|$.
\end{proof}

\cref{thm:vcn} now straightforwardly follows from \cref{lem:algoruntime,lem:algocorrect}.

\section{Conclusion}
In this work, we unify and strengthen several previous NP-hardness results for \textsc{Minimum Temporal Spanner}. Furthermore, we initiate the study of the parameterized complexity of the problem, thereby providing the first non-trivial algorithm for \textsc{Minimum Temporal Spanner} on happy temporal graphs. Our work spawns several natural future research directions. In particular, we leave the following questions open.
\begin{enumerate}
    \item Can the XP-algorithm for \textsc{Minimum Temporal Spanner} parameterized by the vertex cover number be generalized to the non-happy setting?
    %\item Can the XP-algorithm for \textsc{Minimum Temporal Spanner} parameterized by the vertex cover number be improved to an FPT-algorithm?
    \item Can we obtain an XP-algorithm for \textsc{Minimum Temporal Spanner} for a smaller parameter, such as the feedback vertex number?
    %\item Can the W[1]-hardness for \textsc{Minimum Temporal Spanner} parameterized by the feedback vertex number be adapted to show a similar result in the happy setting?
\end{enumerate}

\bibliographystyle{abbrvnat}
\bibliography{bib}	

\end{document}